\documentclass[leqno,12pt]{amsart}
\usepackage[colorlinks,pagebackref,hypertexnames=false]{hyperref}

\usepackage{amsmath,amsthm,amssymb,mathrsfs} 
\usepackage[alphabetic,backrefs]{amsrefs}
\usepackage{ae,aecompl}
\usepackage[margin=1.2in]{geometry} 
\usepackage[onehalfspacing]{setspace}

\usepackage[english]{babel}

\usepackage{bookmark}

\numberwithin{equation}{section}
\usepackage{microtype}

\usepackage{fancybox,fancyhdr,graphics,epsfig}

\numberwithin{equation}{section}

\renewcommand{\det}{\mathop\mathrm{det}\nolimits}

\renewcommand{\epsilon}{\varepsilon}

\newcommand{\tr}{\mathop{\mathrm{tr}}\nolimits}

\def\<{\mathopen{}\left<}
\def\>{\right>\mathclose{}}
\def\({\mathopen{}\left(}
\def\){\right)\mathclose{}}

\newtheorem{conclusion}{Conclusion}

\newtheorem{corollary}{Corollary}

\newtheorem{definition}{Definition}
\newtheorem{example}{Example}

\newtheorem{lemma}{Lemma}

\newtheorem{proposition}{Proposition}
\newtheorem{remark}{Remark}
\newtheorem{question}{Question}
\newtheorem{disclaimer}{Disclaimer}

\numberwithin{equation}{section}

\author{Gon\c calo Oliveira}
\address{Universidade Federal Fluminense}
\email{galato97@gmail.com}

\title{Refined compartmental models, asymptomatic carriers and COVID-19}

\date{April 2020}

\begin{document}
	
	\begin{abstract}
		The goal of this article is to analyze some compartmental models specially designed to model the spread of a disease whose transmission has the same features as COVID-19. %This includes accounting for contagious exposed individuals that are in the incubation period and thus do not (yet) show symptoms; and asymptomatic carriers which are contageous but show mild or no symptoms at all. 
		The major contributions of this article are: (1) Rigorously find sufficient conditions for the outbreak to only have one peak, i.e. for no second wave of infection to form; (2) Investigate the formation of other waves of infection when such conditions are not met; (3) Use numerical simulations to analyze the different roles asymptomatic carriers can have. We also argue that dividing the population into non-interacting groups leads to an effective reduction of the transmission rates.\\
		As in any compartmental model, the goal of this article is to provide qualitative understanding rather than exact quantitative predictions.
	\end{abstract}

\maketitle

\tableofcontents

%===============================================================================

%===============================================================================

%===============================================================================

\section{Introduction}

\subsection{Context}

Compartmental models, based on ordinary differential equations, have a long history of use in epidemiology. Indeed, it is now almost a century since William Ogilvy Kermack and Anderson Gray McKendrick introduced the well known SIR model, \cite{MK}, on which most compartmental modeling elaborates on. Their relevance stems for their extreme simplicity and ability to capture important qualitative behaviors, rather than their limited capacity to make quantitative predictions. As a matter of fact, the current outbreak of COVID-19, caused by SARS-CoV-2, is teaching us that most models, not only the compartmental ones, are of very limited quantitative predictive capacity. Nevertheless, there are important qualitative features of outbreaks that one may attempt to capture, for instance: under which conditions can one guarantee that there will be only one peak? This stands as a very important question which ought to be answered if one is to effectively prevent a second, or third wave of infections. Specially, having in mind the concerns to be taken once a first wave have passed. The goal of this article is to propose and analyze some compartmental models which have been constructed to capture some important features of diseases such as COVID-19. Its two major characteristics which are not captured by the standard compartmental models are:
\begin{itemize}
	\item[(i)] The long incubation period, during which the individuals that have been exposed do not yet develop symptoms but are already contagious (at least in the second half of that period). Even though the standard compartmental models do not take into account the contagiousness of the exposed individuals, they are easy to modify in order to account for it. 
	\item[(ii)] There is a non-negligible fraction of the contagious population which either has mild symptoms or never develops them, thus passing undetected. These are the so-called asymptomatic carriers. After some period of time they are no longer contagious and, as the infected individuals, developed anti-bodies. This suggests they are, at least for now, immune. The standard compartmental models which have a carrier mode are not suitable for modeling such a behavior.
\end{itemize}

In this article, namely in Sections \ref{sec:SEIR} and \ref{sec:SEIAR}, we analyze some compartmental models that have been specially designed to handle these two features. The resulting compartmental models are more complicated than the standard ones and using a mix of precise and numerical results, we shall answer some fundamentally important questions which we now pose:

\begin{question}\label{que:1}
	Regarding each of the models:
	\begin{itemize}
		\item[(a)] Are there disease-free equilibrium states?
		\item[(b)] Having answered positively to (a), one may ask whether there are necessary or sufficient conditions for these disease free equilibrium states to be stable in some sense.
	\end{itemize}
\end{question}

\begin{question}\label{que:2}
	Again, for each of the models one can ask:
	\begin{itemize}
		\item[(a)] Are there criteria guaranteeing that any outbreak of the epidemic will have a unique peak? 
		\item[(b)] If these criteria are violated, can the outbreak have other separate peaks, i.e. second, third waves and so? 
		\item[(c)] Can these possible later peaks be higher than the first?
	\end{itemize}
\end{question}

Depending on which model one considers, these two questions will be answered separately in Sections \ref{sec:SEIR} and \ref{sec:SEIAR}. For instance, using the more refined model we give in Corollary \ref{cor:One_Peak} and Remark \ref{rem:After_Cor} a quantitative condition ensuring only one peak will form. The broad conclusion is that if all transmission rates are kept sufficiently small, then only one peak form, which answers item (a) of question \ref{que:2}. If this condition is not met, more peaks can form as we show in example \ref{ex:Second_Wave_Second_Model} and these can be higher, thus answering items (b) and (c) of question \ref{que:2} above.\\
As for the next question, it will not be answered in any precise way but we shall give some evidence on how to answer it based on numerical results. This will be done in the examples \ref{ex:Realistic}, \ref{ex:Realistic2} and \ref{ex:Second_Wave_Second_Model} presented in subsection \ref{ss:Controling_Outbreak_A}.

\begin{question}\label{que:3}
	What is the role of asymptomatic carriers? Can they help the spread of the disease and thus make the outbreak worse by anticipating and increasing its peak and/or can they shield the rest of the population by creating herd immunity.
\end{question}

One other major question which one poses now is related to the quantifying by how much specific policies reduce the rate of transmission. 

\begin{question}\label{que:4}
	Is there a strategy to effectively reduce the rate of transmission in a manner that can be quantified?
\end{question} 

Broadly speaking, for the general strategy or policy it is hard, if not impossible, to make such a quantitative analysis. For instance, how can we predict by how much the use of masks by the general population reduces the rate of transmission? or washing hands? Indeed, it is impossible to say it a-priori. We only know that these measures do reduce the rate of transmission but not by how much. Having this in mind, we propose one further measure whose contribution to the reduction of the transmission rate can be quantified using the models. The ``idea'' is to split the population into $n$ different groups which are supposed to never interact, this reduces the contact rate of any individual by a factor of $n$ thus reducing the overall transmission rates by a factor of $n$. This intuitive argument can be made more precise using the models we use. The last section of this article, Section \ref{sec:Groups}, is dedicated to making such an argument and running a simulation. We shall also explain how this relates to our answer to the questions \ref{que:1} and \ref{que:2} above.

\subsection{A note on the timing of this article}

Once we are through the worst part of the first wave of the COVID-19 outbreak with relative success, meaning that only a small fraction of the population got infected and so we cannot count with group immunity, the question we must pose is: How can we proceed and go back to a relative normality so that a putative second wave is as mild as possible?\\
I hope that the answers to question \ref{que:2} above, stated as Proposition \ref{lem:Only_One_Maximum} and as Corollary \ref{cor:One_Peak} depending on the model, may serve as good indication that it is possible, but difficult, to have only one peak. They quantify and show what one could already suspect, that the way to do so is to reduce all the transmission rates. It is enough that one transmission rate be high for the outbreak to be out of control, see Conclusion \ref{conclusion:Second_Model}. How to keep the transmission rates small remains a formidable challenge which I do not attempt at answering.\\
As said before, there is no way to quantify how most policies reduce the transmission rates. In that direction Section \ref{sec:Groups} explains a strategy, which may be difficult to implement, but whose effect on the transmission rate can be pined down.

\begin{disclaimer}
	I understand that the strategy of splitting the population into non-interacting groups is of difficult implementation. I do believe it cannot be literally implemented in a society which respects our civil liberties and so would remain at the level of a request to the population, which one would hope to be understood. Of course, in that scenario the different groups would interact, even if mildly. Having such considerations into account leads to modifications of the proposed solutions which will consequently modify the models one should use into mathematically more intractable ones. We may hope that the increased mathematical sophistication of those models will not lead to substantially different qualitative outcomes but instead quantitative ones. In any case, one must understand that the models are simply that: attempts to model a reality which is much more complex. They are, by no means, a truthful reflection of the real world.
\end{disclaimer}

\subsection*{Acknowledgment}
	I want to thank \'Alvaro Kruger Ramos for his comments on an earlier version of this manuscript.\\ 
	Gon\c{c}alo Oliveira is supported by Funda\c{c}\~ao Serrapilheira 1812-27395, by CNPq grants 428959/2018-0 and 307475/2018-2, and by FAPERJ through the grant Jovem Cientista do Nosso Estado E-26/202.793/2019.

\section{A discussion of possible compartmental models}

Recall that the main goal of this article is to find a good and simple compartmental models which realistically capture the properties of a disease similar to COVID-19, at least in terms of the dynamics of transmission. We intend to capture the following:

\begin{itemize}
	\item[(i)] Account for the contagiousness of exposed individuals which may have a long incubation period;
	\item[(ii)] Account for asymptomatic carriers;
	\item[(iii)] Account for the possibly non-negligible mortality rate associated with the disease. 
\end{itemize} 

In fact, looking at a time scale of a few months it is actually conceivable that the background birth and death dynamics be negligible when compared with the death rate associated with the disease itself. Thus, in such a time scale, it is perhaps more realistic to only associate a death rate to the infected population. For an introduction to the use of compartmental models in epidemiology se for example \cite{C} and \cite{AR}.

\subsection{The modified SEIR model}

We shall start by constructing a version of the SEIR model where $S$, $E$, $I$, $R$ represent the fraction of the population which is susceptible, exposed, infected and recovered respectively. In order to keep the system as simple as possible one may attempt to regard the asymptomatic carriers as being part of the exposed population. For this one must regard each exposed individuals as representing an average of someone that is in the incubation period and an asymptomatic carrier. 
We are then led to the following system
\begin{align}\label{eq:ODE1}
\dot{S} & =  - \beta_i I S - \beta_e E S  \\ \label{eq:ODE2}
\dot{E} & = \beta_i I S + \beta_e E S - a E \\ \label{eq:ODE3}
\dot{I} & = a E - \gamma I - \mu I \\ \label{eq:ODE4}
\dot{R} & = \gamma I ,
\end{align}
for some positive functions of time $\beta_e$, $\beta_i$, $a$, $\gamma$. The parameters $\beta_e$, $\beta_i$ represent the transmission rates for the exposed and infected respectively, $a$ denotes the rate of transition from the exposed to infected state and $\gamma$ the rate of recovery. For example, in the easiest case, one would take
$$\beta_i = \frac{\overline{\beta_i}}{S(t)+E(t)+I(t)+R(t)}, \ \text{and} \ \beta_e = \frac{\overline{\beta_e}}{S(t)+E(t)+I(t)+R(t)}, $$
for some constants $\overline{\beta_e}$, $\overline{\beta_i}$. This is not the standard SEIR model as we are forcefully making the exposed interact with the susceptible as the first of these are assumed to be contagious. The analysis of this model is the scope of Section \ref{sec:SEIR}.

\subsection{The SEIAR model}\label{ss:SEIAR}

There are at least two easy ways to further refine this model. First, we may split the exposed into two groups, those which are not yet infectious and those which already are. Secondly, we may have some of these exposed passing directly to the recovered state to account for the asymptomatic carriers which never develop symptoms and have therefore passed undetected. Thirdly, we may distinguish the asymptomatic carriers in a more effective way as they may remain so for a longer period of time. %Of course, some of these may be detected and isolated but I will ignore that.

\subsubsection{Splitting the exposed into groups (SEEIR)}

We may attempt at splitting the exposed into two new compartments, the initially exposed $E_i$ and the final exposed state $E_f$. The idea is that the initially exposed ones are not yet contagious while the final are. To account for the first two requirements above we propose the following system. 

\begin{align*}
\dot{S} & =  - \beta_i I S - \beta_e E_f S  \\
\dot{E}_i & = \beta_i I S + \beta_e E_f S - a_i E_i \\
\dot{E}_f & = a_i E_i - a_f E_f - \gamma_e E_f \\
\dot{I} & = a_f E_f - \gamma_i I - \mu I \\
\dot{R} & = \gamma_i I + \gamma_e E_f.
\end{align*}

I expect this model to be quite good at capturing many qualitative the phenomena. However, given the interest in explicitly analyzing the asymptomatic carriers we shall instead focus on a different model.  %Notice that for any $S=S_c$, $R=R_c$ with $S_c+R_c \leq 1$ and $E_i=E_f=I=0$ it has a disease free equilibrium. To check its stability we can proceed as before linearizing the system above. 

\subsubsection{Having the asymptomatic carriers separate (SEAIR)}\label{subsubsec:Asymptomatic}
The goal of this model is to better capture the effects of asymptomatic carriers. In the case of COVID-19 we do know this to be a non-trivial part of all those which are contagious. However, given the small number of anti-bodies tests that have been made so far, not much quantitative knowledge regarding these is known.\\
In order to account for the asymptomatic carriers we consider one more compartment of the population, denoted by $A$, which represents the fraction of the population which is an asymptomatics carrier. Then, we propose the following model

\begin{align*}
\dot{S} & =  - \beta_i I S - \beta_e E S - \beta_a A S \\ 
\dot{E} & = \beta_i I S + \beta_e E S + \beta_a A S - a_i E - a_a E \\ 
\dot{A} & = a_a E - \gamma_a A \\ 
\dot{I} & = a_i E - \gamma_i I - \mu I \\
\dot{R} & = \gamma_i I + \gamma_a A ,
\end{align*}

for some positive functions of time $\beta_i$, $\beta_e$, $\beta_a$, $a_i$, $a_a$, $\gamma_i$, $\gamma_a$, $\mu$ which need not be constants. This system will be carefully analyzed in Section \ref{sec:SEIAR}. For now, we shall simply comment on the biological meaning of all the parameters on the model. Similar remarks hold for the parameters of the simpler SEIR type model. The $\beta_i$, $\beta_e$, $\beta_a$ respectively represent the transmission rates of the infected, exposed (in th incubation period), and asymptomatic individuals. The parameters $a_i$, $a_a$ represent the rate at which the exposed individuals respectively pass to the infected and asymptotic compartments. In particular, $\tfrac{a_i}{a_i+a_a}$ and $\tfrac{a_a}{a_i+a_a}$ respectively represent the fraction of the population which if exposed will eventually become infected or asymptomatic. The quantity $\gamma_i^{-1}$ represents the average time an infected individual takes to recover, while $\gamma_a^{-1}$ is the average time an asymptomatic individual takes to stop being contagious. Finally, $\mu$ stands for the death rate associated with the disease. For example, in the easiest case one would take 
$$\beta_i = \frac{\overline{\beta_i}}{S(t)+E(t)+A(t)+I(t)+R(t)}$$
for a constant $\overline{\beta_i}$ with similar formulas for $\beta_e$ and $\beta_a$. For convenience, we show in figure \ref{fig:Diagram} a diagrammatic presentation of this model.
\begin{figure}[h]
	\centering
	\includegraphics[width=0.5\textwidth,height=0.3\textheight]{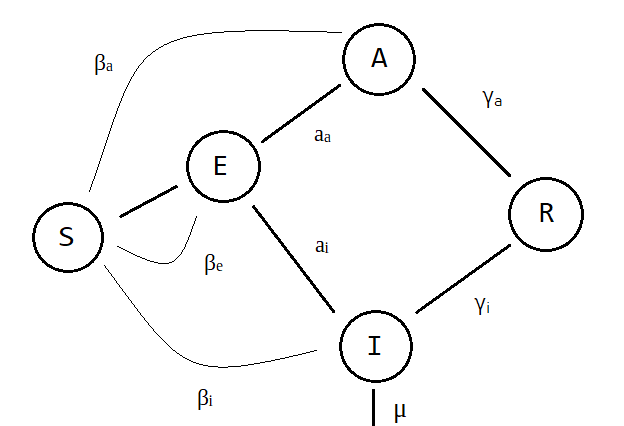}
	\caption{Diagram of the SEIAR model.}
	\label{fig:Diagram}
\end{figure}
A somewhat similar model have also recently been used in \cite{Re} and in \cite{Vo}. The first of these attempted at predicting the fraction of undocumented infections for COVID-19 in China. One other similar version of it using Markov Chains have also recently been used to predict the spatiotemporal spread of COVID-19 in Spain, see \cite{Ae} for this work.

\subsection{Refining the models}

There are a few natural ways to further refine the models proposed above.

\subsubsection{With exposed split and asymptomatic carriers separate (SEEAIR)}\label{subsubsec:Combined}

This is a combination of the previous two models proposed in subsection \ref{ss:SEIAR}. This models also splits the exposed population in two groups and includes the asymptomatic carriers as well.

\begin{align*}
\dot{S} & =  - \beta_i I S - \beta_e E_f S - \beta_a A S \\ 
\dot{E}_i & = \beta_i I S + \beta_e E_f S + \beta_a A S - a E \\
\dot{E}_f & = a E_i - a_f E_f - a_a E_f \\ 
\dot{A} & = a_a E - \gamma_a A \\ 
\dot{I} & = a_i E - \gamma_i I - \mu I \\ 
\dot{R} & = \gamma_i I + \gamma_a A .
\end{align*}

This differs from the previous model by the splitting $E=E_i+E_f$ with the $E_i$ not yet infectious. The $E_f$ may then become visibly infected ($I$) or only exhibit mild symptons or even completely asymptomatic ($A$).

\subsubsection{With birth and death dynamics}

We can further refine the previous model by including a birth and death dynamics. We assume there is a birth rate $\Lambda$ of healthy individuals and a normal death rate $\mu_n$ by other conditions totally unrelated to the disease which affects all the population. In this manner the model becomes

\begin{align*}
\dot{S} & = \Lambda - \beta_i I S - \beta_e E_f S - \beta_a A S - \mu_n S \\ 
\dot{E}_i & = \beta_i I S + \beta_e E_f S + \beta_a A S - a E - \mu_n E_i \\
\dot{E}_f & = a E_i - a_f E_f - a_a E_f - \mu_n E_f \\ 
\dot{A} & = a_a E_f - \gamma_a A - \mu_n A \\ 
\dot{I} & = a_i E_f - \gamma_i I - \mu I -\mu_n I\\ 
\dot{R} & = \gamma_i I + \gamma_a A -\mu_n R.
\end{align*}

As we are interested in running the model for short periods of time, such as a few months, I expect the birth rate to be relatively small and so negligible. Similarly, as the normal death rate equally affects the whole populating I believe that including it in the model will simply lead to unimportant complications.

%
%\subsubsection{Generalized SEEIR type models}
%
%Generalizing all the above we consider the most general kind of a SEEIR. This splits each compartment into sub-compartments, so for instance we could consider the asymptomatics and infected as different sub-compartments of a larger infected compartment. Then, we write 
%$$S=\sum_{k=1}^{n_s} S^{(k)}, \ E_i=\sum_{k=1}^{n_{e_i}} E_i^{(k)}, \ E_f = \sum_{k=1}^{n_{e_f}} E_f^{(k)}, \ I=\sum_{k=1}^{n_i} I^{(k)}, \ R=\sum_{k=1}^{n_r} R^{(k)},$$
%and one can consider the following system of ODE's
%\begin{align*}
%\dot{S}^{(k)} & = \Lambda_k  \left( \sum_{l=1}^{n_{e_f}} \beta_e^{(k,l)} E_f^{l}  + \sum_{l=1}^{n_i} \beta_i^{(k,l)} I^{(l)}  \right) S^{(k)} - \mu_n S^{(k)} \\ 
%\dot{E}_i & = \beta_i I S + \beta_e E_f S + \beta_a A S - a E - \mu_n E_i \\
%\dot{E}_f & = a E_i - a_f E_f - a_a E_f - \mu_n E_f \\ 
%\dot{A} & = a_a E_f - \gamma_a A - \mu_n A \\ 
%\dot{I} & = a_i E_f - \gamma_i I - \mu I -\mu_n I\\ 
%\dot{R} & = \gamma_i I + \gamma_a A -\mu_n R.
%\end{align*}
%

\section{The simplest somewhat realistic model}\label{sec:SEIR}

The goal of this section is to analyze the model proposed in subsection \ref{sec:SEIR}. For convenience we recall here the system of ordinary differential equations ruling it.

\begin{align}\label{eq:ODE1}
\dot{S} & =  - \beta_i I S - \beta_e E S  \\ \label{eq:ODE2}
\dot{E} & = \beta_i I S + \beta_e E S - a E \\ \label{eq:ODE3}
\dot{I} & = a E - \gamma I - \mu I \\ \label{eq:ODE4}
\dot{R} & = \gamma I .
\end{align}

It is straightforward to check that all $S,E,I,R$ remain nonnegative and that $S+E+I+R $ is non-increasing and so $S+E+I+R \leq S(0)+E(0)+I(0)+R(0) =1$.

\subsection{The linearized system}
From the last equation above, i.e. \ref{eq:ODE4}, a critical point 
$$(S,E,I,R)=(S_c,E_c,I_c,R_c)$$ 
will certainly have $I_c=0$ which then also gives $E_c=0$ from the semi-last equation \ref{eq:ODE3}. Then, the reaming equations \ref{eq:ODE1} and \ref{eq:ODE2} we conclude that $R_c$ and $S_c$ can be any two constants so that $R_c+S_c \leq 1$. These are disease free equilibrium points and are obviously not isolated. Thus, the Grobman-Hartmann theorem cannot be applied but I believe the linearization to still carry important information. Thus, we shall linearize the system around such an equilibrium and use it to define the following notion.

\begin{definition}
	Given a disease free equilibrium point $(S_c , 0, 0, R_c)$ is called infectiously-stable if the associated linear system is such that any of its solutions $(s,e,i,r)$ satisfies
	$$\lim_{t \to + \infty} i(t)=0.$$
\end{definition}

Intuitively, we may think of this condition as saying that perturbing away from the equilibrium always makes the disease to get extinct. In what follows of this linear analysis we shall suppose $\beta_i$, $\beta_e$ are constant which is a reasonable assumption for large $t \gg 1$. In order to study their stability, we shall now linearize the equations around these equilibria. The linearized system is given by
\begin{align*}
\dot{s} & = - \beta_i S_c i - \beta_e S_c e  \\
\dot{e} & = \beta_i S_c i + \beta_e S_c e  - a e \\
\dot{i} & = a e - (\gamma +\mu ) i \\
\dot{r} & = \gamma i .
\end{align*}
Clearly, the dynamics of $(e,i)$ controls the whole dynamics which evolve independently. The stability is therefore dependent on the sign of the eigenvalues. As we want both of these two be negative we need to have $\det(A)>0$ and $\tr(A)<0$, where $A$ is the matrix 
$$A= \begin{pmatrix}
\beta_e S_c - a & \beta_i S_c  \\
a & - (\gamma +\mu )
\end{pmatrix},$$
which controls the system. Thus, the condition for stability is
\begin{align*}
- (\gamma +\mu )(\beta_e S_c - a) - a \beta_i S_c & > 0 \\
\beta_e S_c - (\gamma +\mu ) & < 0,
\end{align*}
which we can rewrite as
\begin{align*}
a \beta_i S_c & < (\gamma +\mu )(a-\beta_e S_c) \\
\beta_e S_c & < (\gamma +\mu ) .
\end{align*}
The second of this equations says that the rate at which the exposed infect other people must be smaller than the rate of recovery (plus the death rate). The first says that not only the exposed must be infecting people slower than they pass to the infected state ($\beta_eS_c < a$) but also, the already infected must infect others at a very very slow rate when compared to the recovery rate. More precisely, we have obtained the following result.

\begin{lemma}
	For the generic values of $\beta_e$, $\beta_i$, $a$, $\gamma$, $\mu$. The disease free equilibrium $(S_c , 0, 0, R_c)$ is infectiously-stable if 
	\begin{equation}\label{eq:Upper_Bound_S_c}
	S_c < \min \big\lbrace \frac{\gamma+\mu}{\beta_e} , \frac{a(\gamma+ \mu)}{a \beta_i + (\gamma+\mu) \beta_e} \big\rbrace .
	\end{equation}
	In particular, if both $\frac{\gamma+\mu}{\beta_e}$ and $\frac{a(\gamma+ \mu)}{a \beta_i + (\gamma+\mu) \beta_e}$ are greater than one, we find that any disease free equilibrium is stable.
\end{lemma}

\begin{remark}\label{conclusion:First_Model}
	It is possible to reduce the population of those infected keeping many people $S_c<S(0)$ still susceptible (without ever having been infected). One way to do that is to have the transmission rate of those exposed small when compared to the rate of recovery, i.e.
	$$\frac{\beta_e}{\gamma + \mu}  < 1 .$$
	Furthermore, one must have the rate of transmission of those infected even smaller when compared with rate of recovery, namely that
	$$ \beta_i< (\gamma + \mu)\left(1-\frac{\beta_e}{a}\right) , $$
	which in particular also assumes that $\beta_e < a$, i.e. that the rate at which the exposed infect the susceptible is small when compared with the rate of transition from the exposed to infected state. We may also rewrite this last condition as
	$$\frac{\beta_i}{\gamma+\mu} + \frac{\beta_e}{a} < 1,$$
	which suggests regard the left-hand-side as an analogue of the basic reproduction rate for this model.\\
	Achieving these conditions may, however, be a difficult task to accomplish. If they are not, then the final population of susceptible may be very low which is quite bad.	
\end{remark}

\subsection{Global stability}\label{subsec:Global_Stability}

We shall now find a sufficient condition for a solution to the system \ref{eq:ODE1}--\ref{eq:ODE4} to asymptotically approach a disease free equilibrium state. With this in mind, it is convenient to abstract such a notion as follows. 

\begin{definition}
	A solution $(S,E,I,R)$ to \ref{eq:ODE1}--\ref{eq:ODE4} is called asymptotically disease free if 
	$$\lim_{t \to + \infty} E(t)+I(t) =0 .$$
\end{definition}

\begin{proposition}
	Let $(S,E,I,R)$ be a solution to \ref{eq:ODE1}--\ref{eq:ODE4} satisfying 
	\begin{equation}\label{eq:assumption0}
	\sup_{t \geq 0}\left( \frac{\beta_i}{\gamma+\mu} \right)  + \sup_{t \geq 0}\left( \frac{\beta_e}{a} \right)  < 1.
	\end{equation}
	Then, $(S,E,I,R)$ is asymptotically disease free.
\end{proposition}
\begin{proof}
	Let $\epsilon \in (0,1)$ to be fixed later and consider the function $L_{\epsilon}(t)= E(t) + \epsilon I(t)$. Then, using equations \ref{eq:ODE2} and \ref{eq:ODE3} we compute
	\begin{align*}
	\dot{L_\epsilon} & = \dot{E} + \epsilon \dot{I} \\
	& = (\beta_eES + \beta_i I S -aE) + \epsilon( a E - (\gamma+\mu)I) \\
	& = a E \left( \frac{\beta_e}{a} S - (1-\epsilon) \right) + (\gamma+\mu) I \left( \frac{\beta_i}{\gamma+\mu} S - \epsilon \right) .
	\end{align*}
	Then, by the assumption \ref{eq:assumption0} there is $\epsilon \in (0,1)$ such that $\sup_{t \geq 0}\left( \frac{\beta_i}{\gamma+\mu}\right) S(0) <  \epsilon$ and $\sup_{t \geq 0}\left(  \frac{\beta_e}{a} \right) S(0) < 1 - \epsilon$. Picking such an $\epsilon$, the above computation shows that $L_{\epsilon}$ is decreasing and satisfies an inequality of the form
	$$\dot{L_\epsilon} \leq - \delta L_{\epsilon},$$
	for some $\delta >0$. Then, Gr\"onwall's inequality yields that $L_{\epsilon} \leq L_{\epsilon}(0) e^{-\delta t}$ and thus converges to $0$ as $t \to \infty$. In particular, as both $E$ and $I$ are non-negative we find that both of these must converge to $0$.
\end{proof}

\begin{remark}
	The previous proof actually shows that it is enough that 
	$$\sup_{t \geq 0}\left( \frac{\beta_i}{\gamma+\mu} S(t) \right)  + \sup_{t \geq 0}\left( \frac{\beta_e}{a} S(t) \right)  < 1,$$
	for the same conclusion to hold. Furthermore, it shows that under the above hypothesis, both $E$ and $I$ exponentially converge to $0$. 
\end{remark}

\subsection{Examples and numeric simulations}

\begin{example}\label{ex:First}
	Consider this model with $S(0)=0.99$, $I(0)=0$, $E(0)=0.01$ and $R(0)=0$ (for small nonzero $I(0),E(0)$ the outcome seems to not be very dependent on the precise value) whose simulation is shown in figure \ref{fig:First}. In this precise case, only a very small fraction of the population got away without being infected. %We do not want to be in this scenario.
	\begin{figure}[h]
	\centering
	\includegraphics[width=0.4\textwidth,height=0.25\textheight]{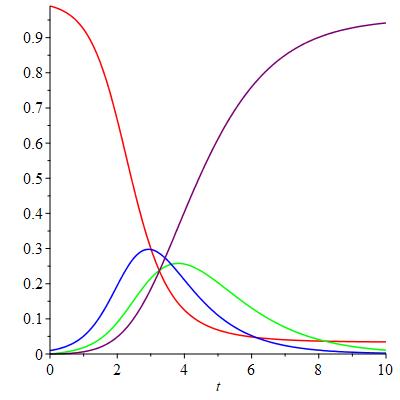}
	\caption{Example when $\beta_i=0.9$, $\beta_e=2.5$, $a=1$, $\gamma=0.9$, $\mu=0.1$ which has 	
		$$\frac{\beta_e}{\gamma+\mu} \sim 2.78 , \ \ \frac{\beta_i}{\gamma+\mu}+ \frac{\beta_e}{a} \sim 3.5,$$ 
		and initial conditions $S(0)=0.99$, $I(0)=0$, $E(0)=0.01$ and $R(0)=0$. The red, blue, green and purple curves respectively denote the fraction of the population composed of susceptible, exposed, infected and recovered.}
	\label{fig:First}
	\end{figure}
	Indeed, form equation \ref{eq:Upper_Bound_S_c} for the stability of the disease free equilibrium we find that it must satisfy
	\begin{equation}
	S_c < \min \big\lbrace \frac{\gamma+\mu}{\beta_e} , \frac{a(\gamma+ \mu)}{a \beta_i + (\gamma+\mu) \beta_e} \big\rbrace =\frac{2}{7} ,
	\end{equation}
	which is in agreement with what we see in figure \ref{fig:First}.\\
	On the other hand, we may consider the same model but supposing that social isolation and other measures have been put in place to reduce the transmission rate. In figure \ref{fig:First2} we run one such example.
	\begin{figure}[h]
		\centering
		\includegraphics[width=0.4\textwidth,height=0.25\textheight]{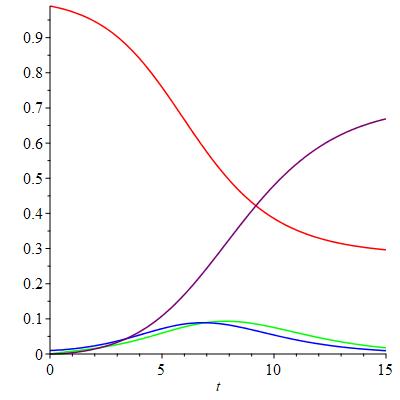}
		\caption{Example when $\beta_i=0.45$, $\beta_e=1.25$, $a=1$, $\gamma=0.9$, $\mu=0.01$ which has 	
			$$\frac{\beta_e}{\gamma+\mu} \sim 1.389 , \ \ \frac{\beta_i}{\gamma+\mu}+ \frac{\beta_e}{a} \sim 1.75,$$
			and initial conditions $S(0)=0.99$, $I(0)=0$, $E(0)=0.01$ and $R(0)=0$. The color code is that of figure \ref{fig:First}.}
		\label{fig:First2}
	\end{figure}
	In this case, equation \ref{eq:Upper_Bound_S_c} gives that the stable disease free equilibrium to which the model converges satisfies $S_c<\tfrac{4}{7}$ which is compatible with what can be inferred from figure \ref{fig:First2}.
\end{example}

\begin{remark}
	It would be interesting to actually get a lower bound on $S_c$ rather than an upper bound.
\end{remark}

\subsubsection{Example with two peaks}

We could now try and make $\beta_i$ and $\beta_e$ periodic in a way that, on average, the inequalities in Conclusion \ref{conclusion:First_Model} hold true. In that case, I expect that, after the peak, the number of infected people will also decrease, in average, with time. Of course, it may be that, with no extra measures, the actual values of $\beta_e$ and $\beta_i$ are so high that by opening up (even if only for a short period of time) will make the task of making the inequalities above hold in average very difficult. It is exactly this scenario that we explore in the next example where after a strict control of the transmission rate, it is once again allowed to become large. This is an alert to the dangers of a non-careful reopening, after social distancing measures had been put in place, for an insufficient time so that herd immunity develops. In this example we can actually see that the second peak can be made larger than the first which provides an answer to the items (b) and (c) in question \ref{que:2}.

\begin{example}\label{ex:Second_Peak_Easy_Model}
	Figure \ref{fig:Second_Peak} shows an example where after seeing positive signs in the decrease in the number of infected, the rate of transmission is allowed to increase. This leads to a formation of a second peak of infections which ends up infecting almost everyone.\footnote{This model is made with $\gamma=0.9$, $\mu=0.01$, $a=1$ and $\beta_i=0.9(\tfrac{H(1-t)}{10}+H(t-5))$, $\beta_e=2.5(\tfrac{H(1-t)}{10}+H(t-5))$, where $H(t)=\tfrac{1}{10}+S_{23}(t)$ for $S_{23}$ the first 23 terms of the Fourier series in $[-20,20]$ of the Heaviside function.} 
	\begin{figure}[h]
		\centering
		\includegraphics[width=0.4\textwidth,height=0.25\textheight]{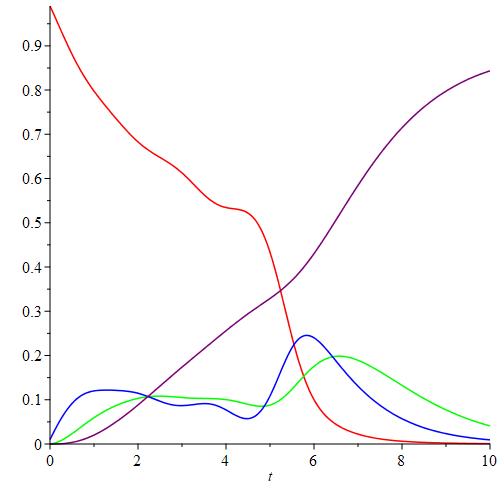}
		\caption{Example of an un-careful reopening which leads to a second peak.}
		\label{fig:Second_Peak}
	\end{figure}
	Indeed, we can see in the plot in figure \ref{fig:Criteria_1} that the conditions in Conclusion \ref{conclusion:First_Model} are violated right before time $t=5$. This leads to a change of tendency which will makes the curve of infected restart increasing and eventually leading to the formation of the second peak. This is an example of something we want to avoid.
	\begin{figure}[h]
		\centering
		\includegraphics[width=0.4\textwidth,height=0.25\textheight]{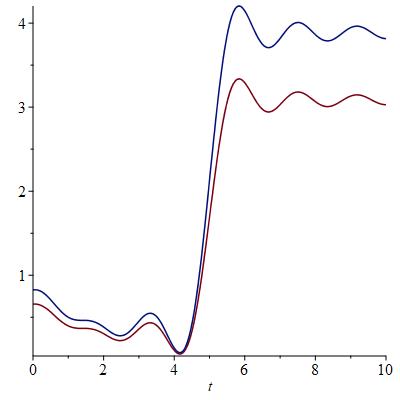}
		\caption{Plots of $\frac{\beta_e}{\gamma+\mu}$ and $\frac{\beta_i}{\gamma+\mu} + \frac{\beta_e}{a}$ which one can see to surpass the value $1$ before $t=5$.}
		\label{fig:Criteria_1}
	\end{figure}
\end{example}

\begin{remark}
	This example leads to a fundamental question. Is there a criteria which guarantees that no second peak will form? This is precisely the problem raised in th problem (a) by question \ref{que:2}. Preferably a criteria depending on $\beta_e$ and $\beta_i$ which are the quantities that can be somewhat controlled using social distancing and other measures. We shall answer this question next in Proposition \ref{lem:Only_One_Maximum}.
\end{remark}

\subsection{Rigorous qualitative results}

First, notice that from the first equation \ref{eq:ODE1} we find that $\dot{S} \leq 0$ so $S(t)$ is decreasing. Consider the asymptotic fraction of susceptible population 
$$S_c := \lim_{t \rightarrow + \infty} S(t) , $$
and recall that it is to our interest that $S_c>0$ and as big as possible (of course $S_c \leq S(0) \leq 1$). Similarly, we shall define $I_c$ and $E_c$ as the asymptotic values of $I$ and $E$ if they exist.

\begin{lemma}\label{lem:E_I_to_zero}
	Suppose that $S_c >0$, then $I_c=0=E_c$.
\end{lemma}
\begin{proof}
	Again, from the first equation \ref{eq:ODE1} we have
	$$S(t) = S(0) \exp \left( - \int_0^t (\beta_iI+\beta_eE) ds \right).$$
	which will converge to zero if $I$ and $E$ do not themselves uniformly converge to zero.
\end{proof}

Political measures can more easily affect the transmission rate than the recovery and death rate which are more dependent on medical conditions. Thus, in the next result we find a, potentially useful, criteria for having only one peak which does not assume the transmission rates $\beta_i, \beta_e$ to be constant. This provides a rigorous answer to the item (c) of question \ref{que:2}.

\begin{proposition}\label{lem:Only_One_Maximum}
	Suppose that $a, \gamma, \mu$ are all constant and $(S,E,I,R)$ is a nonconstant solution of \ref{eq:ODE1}--\ref{eq:ODE4}. If there is a time $t_\ast$ such that $\dot{I}(t_\ast)<0$ and
	$$S(t_\ast)\leq \inf_{t \geq t_*} \left( \frac{a (\gamma+\mu)}{a \beta_i + (\gamma + \mu) \beta_e} \right),$$
	then, for $t \geq t_\ast$, the fraction of infected population $I(t)$ has at most one critical point. Moreover, if any such critical point exists, then it is a maximum. In particular, if
	\begin{equation}\label{eq:Criteria_No_Second_Peak}
	\sup_{t \geq 0} \left( \frac{\beta_e}{a} + \frac{\beta_i }{\gamma+\mu} \right) \leq 1,
	\end{equation}
	then $I(t)$ at most a unique maximum.
\end{proposition}
\begin{proof}
	First notice that if $I$ and $E$ both vanish at some point, then the solution is constant. Then, we may suppose this is not the case and differentiate the third equation \ref{eq:ODE3}, using the second \ref{eq:ODE2} to substitute for $\dot{E}$. This gives
	\begin{align*}
	\ddot{I} = a \dot{E} - (\gamma+\mu) \dot{I} = a (\beta_i I S + \beta_e E S- a E) - (\gamma+\mu) \dot{I},
	\end{align*}
	which at a critical point of $I$ yields
	$$\ddot{I} = a (\beta_i I S + \beta_e E S ) - a^2 E . $$
	As between any two maxima there must be a minimum, in order for $I$ to only have one maximum it is enough that it has no minimum. At a minimum we would have $\ddot{I} \geq 0$ which from the previous computation is ruled out if
	$$ a \beta_i I S +  ( \beta_e S  - a ) aE. $$
	However, the condition that we are at a critical point of $I(t)$, i.e. $\dot{I}(t)=0$, yields that $aE=(\gamma+\mu) I$ which upon inserting above gives $ a \beta_i S +  ( \beta_e S  - a ) (\gamma+\mu) <0 $, which in terms of $S$ reads as
	$$S<\frac{a (\gamma+\mu)}{a \beta_i + (\gamma + \mu) \beta_e}.$$
	Thus, if $S$ always satisfies this inequality past some time, then there is at most one critical point of $I$ and this must a maximum, if it exists at all. 
\end{proof}

\begin{remark}
	The previous result gives a criteria to guarantee that a second peak will not form. For that, one must try and control the $\beta_e$ and $\beta_i$ so that equation \ref{eq:Criteria_No_Second_Peak} holds.
\end{remark}

\section{A model capturing asymptomatic carriers}\label{sec:SEIAR}

In this section we analyze the model previously derived in \ref{subsubsec:Asymptomatic} which was designed to capture asymptomatic carriers. This model regards the exposed (E) in a slightly different manner than in the more refined model \ref{subsubsec:Combined}. Namely, we consider the exposed as only one state and average out their infectiousness. For convenience, we recall here the system of equations governing this model.

\begin{align}\label{eq:A1}
\dot{S} & =  - \beta_i I S - \beta_e E S - \beta_a A S \\ \label{eq:A2}
\dot{E} & = \beta_i I S + \beta_e E S + \beta_a A S -a_a E - a_i E \\ \label{eq:A3}
\dot{A} & = a_a E - \gamma_a A \\ \label{eq:A4}
\dot{I} & = a_i E - \gamma_i I - \mu I \\ \label{eq:A5}
\dot{R} & = \gamma_i I + \gamma_a A .
\end{align}

From equation \ref{eq:A1} we find that $S(t)$ is decreasing and from adding \ref{eq:A1} and \ref{eq:A2} we find that $S(t)+E(t)$ as well.

\subsection{Linear analysis}\label{ss:Linear_Second_Model}

Let us now find the critical points of this system. From equations \ref{eq:A3}, \ref{eq:A4} and \ref{eq:A5} we find that for the generic values of $a_a$, $a_i$, $\gamma_i$, $\gamma_a$, $\mu$ all $E,A,I$ must vanish at a critical point. From the remaining equations we find that $S$ and $R$ can be any arbitrary constants as long as $S+R \leq 1$. We shall write such a critical point as $c=(S,E,A,I,R)=(S_c,0,0,0,R_c)$. Notice in particular that any such critical point is disease free.\\
Then, the linearised system at such a critical point is given by
\begin{align*}
\dot{s} & =  - \beta_i S_c i - \beta_e S_c e - \beta_a S_c a \\
\dot{e} & = \beta_i S_c i + \beta_e S_c e + \beta_a S_c a - (a_i + a_a )e \\ 
\dot{a} & = a_a e - \gamma_a a \\ 
\dot{i} & = a_i e - (\gamma_i +\mu) i \\
\dot{r} & = \gamma_i i + \gamma_a a .
\end{align*}
In this case, we have that the equations for $\dot{x}$ where $x:=(e, a, i)$ do control the whole system. This subsystem is then given by $\dot{x}=Ax$ where
$$A = \begin{pmatrix}
\beta_e S_c - (a_i + a_a) & \beta_a S_c & \beta_i S_c \\
a_a & - \gamma_a & 0 \\
a_i & 0 & - (\gamma_i + \mu)
\end{pmatrix}.$$
While it is possible to compute the eigenvalues of this system in general, the formulas are extremely unwieldy and so it is hard to find a general statement which decides on the stability of the system in general. Nevertheless, as necessary condition for stability is that $\det(A)<0$ and $\tr(A)<0$ which respectively give
\begin{align}\label{eq:Necessary_Conditions}
S_c < \frac{\gamma_a (\gamma_i + \mu) (a_i + a_a)}{\gamma_a (\gamma_i + \mu) \beta_e +  \gamma_a a_i \beta_i +  (\gamma_i + \mu) a_a \beta_a}  , \ \text{and} \ 
S_c < \frac{a_i + a_a + \gamma_a + \gamma_i + \mu}{\beta_e} .
\end{align} 
If $\beta_e < a_i +a_a + \gamma_a + \gamma_i + \mu$, then the first equation above suggests that the quantity
\begin{equation}\label{eq:R0_A}
\frac{1}{a_i+a_a} \left( \beta_e + a_i \frac{\beta_i}{\gamma_i + \mu} + a_a \frac{\beta_a}{\gamma_a} \right),
\end{equation}
may behave as the basic reproduction rate for this model. Hence, and in order to have large (closed to $1$) values of the asymptotic fraction of susceptible, we would like to keep \ref{eq:R0_A} below $1$. This is then compatible with having $S_c$ close to $1$, at least from the point of view of the first equation in \ref{eq:Necessary_Conditions}.

\begin{remark}\label{conclusion:Second_Model}
	The above discussion suggests that for an outbreak to be under control we must have $\beta_e < a_i +a_a + \gamma_a + \gamma_i + \mu$ and the quantity in equation \ref{eq:R0_A} must be kept below $1$. Furthermore, from formula \ref{eq:R0_A} one can immediately see that it is enough that one of the transmission rates $\beta_e$, $\beta_i$, or $\beta_a$ to be large, for the outbreak to be out of control!
\end{remark}

\subsection{Global stability}

As in subsection \ref{subsec:Global_Stability} we will now characterize the asymptotic stability properties of the disease free equilibrium points of solutions to \ref{eq:A1}--\ref{eq:A5}. As in that subsection, it is convenient to abstract the notion of asymptotic stability for these equilibrium. 

\begin{definition}
	A solution $(S,E,A,I,R)$ to \ref{eq:A1}--\ref{eq:A5} is called asymptotically disease free if 
	$$\lim_{t \to + \infty} E(t)+A(t)+I(t) =0 .$$
\end{definition}

\begin{proposition}
	Let $(S,E,A,I,R)$ be a solution to \ref{eq:A1}--\ref{eq:A5} satisfying 
	\begin{equation}\label{eq:assumption01}
	\sup_{t \geq 0}\frac{\beta_e}{a_i+a_a}  + 	\sup_{t \geq 0}\frac{a_i}{a_i+a_a} \frac{\beta_i}{\gamma_i + \mu}   + 	\sup_{t \geq 0}\frac{a_a}{a_i+a_a} \frac{\beta_a}{\gamma_a}  < 1 .
	\end{equation}
	Then, $(S,E,A,I,R)$ is asymptotically disease free. Furthermore, all $E$, $A$ and $I$ exponentially converge to $0$.
\end{proposition}
\begin{proof}
	Following a similar strategy to that of subsection \ref{subsec:Global_Stability} we shall consider the function $L_{\epsilon}(t)= E(t) + \epsilon_a A(t)+ \epsilon_i I(t)$ for constants $\epsilon_a , \epsilon_i \in (0,1)$ to be fixed at a later stage. Then, it from equations \ref{eq:A2}, \ref{eq:A3} and \ref{eq:A4} follows that
	\begin{align*}
	\dot{L_\epsilon} & = (\beta_eES + \beta_a A S + \beta_i I S - (a_a+a_i) E) + \epsilon_a ( a_a E - \gamma_a A ) + \epsilon_i ( a_i E - (\gamma_i+\mu)I) \\
	& = (a_i + a_a) E \left( \frac{\beta_e}{a_i+a_a} S - \left(1- \frac{\epsilon_a  a_a + \epsilon_i a_i}{a_i+a_a}  \right) \right) + \gamma_a A \left( \frac{\beta_a}{\gamma_a} S - \epsilon_a \right)  + (\gamma_i+\mu) I \left( \frac{\beta_i}{\gamma_i +\mu} S - \epsilon_i \right) ,
	\end{align*}
	which we can rewrite as
	$$(a_i + a_a) E \left( \frac{\beta_e}{a_i+a_a} S - \left(1- \frac{\epsilon_a  a_a + \epsilon_i a_i}{a_i+a_a}  \right) \right) + \frac{\gamma_a A}{a_a} \left( \frac{a_a\beta_a}{\gamma_a} S - \epsilon_a a_a \right)  + \frac{\gamma_i+\mu}{a_i} I \left( \frac{a_i\beta_i}{\gamma_i +\mu} S - \epsilon_ia_i \right) .
	$$
	By the assumption \ref{eq:assumption01} there are $\epsilon_a, \epsilon_i \in (0,1)$ such that  
	$$	\sup_{t \geq 0} \frac{\beta_a}{\gamma_a}S < \epsilon_a , \ \ \sup_{t \geq 0} \frac{\beta_i}{\gamma_i + \mu}S < \epsilon_i $$
	and $\frac{\beta_e}{a_i+a_a} <1- \frac{\epsilon_a  a_a + \epsilon_i a_i}{a_i+a_a}$. Thus, we find that $L_{\epsilon}$ is decreasing and satisfies
	$$\dot{L_\epsilon} \leq - \delta L_{\epsilon},$$
	for some positive $\delta >0$. As a consequence of Gr\"onwall's inequality we find that $L_{\epsilon} \leq L_{\epsilon}(0) e^{-\delta t}$ and thus converges exponentially to $0$ as $t \to \infty$.
\end{proof}

\begin{remark}
	The previous proof actually shows that the same result holds under the weaker hypothesis that
	$$\sup_{t \geq 0}\frac{\beta_e(t) S(t)}{a_i+a_a}  + 	\sup_{t \geq 0}\frac{a_i}{a_i+a_a} \frac{\beta_i(t) S(t)}{\gamma_i + \mu}   + 	\sup_{t \geq 0}\frac{a_a}{a_i+a_a} \frac{\beta_a(t) S(t)}{\gamma_a}  < 1 .$$ 
\end{remark}

\subsection{Controlling an outbreak}\label{ss:Controling_Outbreak_A}

The main problem when faced with an outbreak is to keep it under control. But what does that mean in the context of the model we are considering. One possible definition is that the asymptotic value os susceptible
$$S_c := \lim_{t \to + \infty}S(t),$$
is large. Alternatively, we may attempt at keeping the total number of infected people, i.e.
$$\lim_{t \to + \infty} \int_0^tI(s) ds,$$
as small as possible. The proof of Lemma \ref{lem:E_I_to_zero} show that these two definitions are equivalent for the model \ref{eq:ODE1}--\ref{eq:ODE4}. However, in model at hand there is the possibility that as susceptible passes to the recovered state taking the asymptomatic route which makes this model fundamentally different. Indeed, from integrating the first equation \ref{eq:A1} we find that
$$S(t) = S(0) \exp \left( -\int_0^t (\beta_iI+\beta_eE+ \beta_a A) ds \right).$$
This shows that maximizing $S_c$ is equivalent to minimizing
$$\int_0^t (\beta_iI+\beta_eE+ \beta_a A) ds.$$
From the point of view of the second definition, one regards the asymptomatic to be innocuous, which seems reasonable from a purely medical point of view as they never become sick. Furthermore, it is conceivable that they can create herd immunity which shields the rest of the population. However, intuitively speaking, having a relatively large number of asymptomatic may also be bad, at least initially while $S$ is large, as they may move undetected and ``infect'' even more people. The numerical conclusion to which we arrive is that depending on the exact parameters of the model the asymptomatic can carry both roles.

\subsubsection{Examples}

Below we shall run some simulations which seem to suggest the previous intuitive reasoning to be correct in different settings. In order to this we shall fix the parameters $\beta_i,\beta_e,\beta_a,\gamma_i,\gamma_a,\mu$ and vary $a_i$ and $a_a$. The examples in this section have been constructed in order to provide answers to the question \ref{que:3} raised in the Introduction.

\begin{remark}
	A few words must be said about the choice of the parameters $\beta_i,\beta_e,\beta_a,a_i,a_a,\gamma_i,\gamma_a,\mu$ ruling the model. It is at the date medically impossible to determine the values of all these and so it is meaningless to assign them specific values and pretend we are modeling the true outbreak. Nevertheless, I must justify the values I will be assigning in my model. I shall suppose that 
	$$\beta_i < \beta_e < \beta_a,$$ 
	and recall that these encode the number of infections that the infected, exposed and asymptomatic cause by unit time. This choice may seem strange but there is a reason for this choice. I assume that the infected, while possibly being extremely infectious, are dealt with carefully and so do not infect as many other people by unit time as the exposed and asymptomatic. This justifies $\beta_i < \beta_a$ and also $\beta_i < \beta_e$. The exposed, even though not having been detected yet, are not as infectious as they will eventually become later thus $\beta_e < \beta_a$. Notice that we do not have $\beta_e<\beta_i$ as the infected individuals will be detected and dealt with in a way that avoids further infections. Without this extra care that inequality would hold. From all this we then have:
	$\beta_i < \beta_e < \beta_a$, as claimed.\\
	Next, we shall assume in our simulations that 
	$$\gamma_i< \gamma_a$$ 
	which means that the average time $\gamma_a^{-1}$ an asymptomatic individual takes to reduce its viral charge to zero is smaller than the time an infected indivual needs $\gamma_i^{-1}$. This is highly debatable and, even though seemingly reasonable, I have no way to better justify this choice.
\end{remark}

The simulation in the next example \ref{ex:Realistic} illustrates a case where having a larger fraction of asymptomatics carriers did not make the disease develop faster and actually lead to a smaller peaks of infection. This need not always be the case as we shall later see in example \ref{ex:Realistic2}.

\begin{example}\label{ex:Realistic}
	In all the simulations we shall run in this example we shall have $\beta_i=0.9$, $\beta_e=1.5$ and $\beta_a=2.8$ while $\gamma_i=0.8$, $\gamma_a=1.2$ and $\mu=0.01$.\\
	In Figure \ref{fig:Comparison_Same_Sum} we shall plot the number of infected people, these are those which actually get sick, as a function of time where we compare different simulations obtained with the same value of $a_i+a_a=1$ but different values of $a_i$ and $a_a$ individually.
	\begin{figure}[h]
		\centering
		\includegraphics[width=0.5\textwidth,height=0.35\textheight]{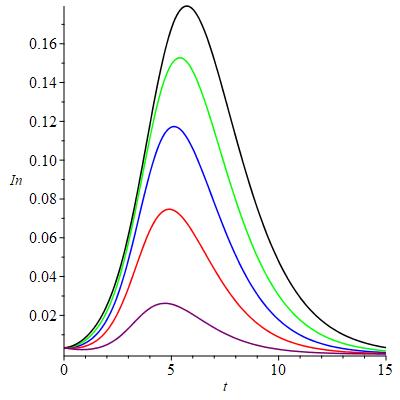}
		\caption{In black, green, blue, red and purple the values of $I(t)$ are plotted for the simulations obtained from $a_a$ being $0.1$, $0.3$, $0.5$, $0.7$ and $0.9$ respectively.}
		\label{fig:Comparison_Same_Sum}
	\end{figure}
	Indeed, from this simulation we can infer that the smaller $a_a$ is, or equivalently the larger $a_i$ is, the further high and late the peak is. Intuitively we can justify this from the fact that a large $a_a$ gives a lot of asymptomatic carriers, which when recovered create a group immunity shielding the rest of the population. 
	%However, we do can see that the green and blue curve do surpass temporarily the black curve before their peaks are achieved, this arises because in that early stage the large fraction of asymptomatic carriers causes a faster increase of the infected population. In order to make this easier to visualize we include in figure \ref{fig:Comparison_Log} a zoomed version around the intersection points in logarithmic scale.
	%\begin{figure}[h]\label{fig:Comparison_Log}
	%	\centering
	%	\includegraphics[width=0.5\textwidth,height=0.35\textheight]{Comparison_Log}
	%	\caption{Highlight of figure \ref{fig:Comparison_Same_Sum} in logarithmic scale.}
	%\end{figure}
\end{example}

I believe the values of the transmission rates $\beta_e$, $\beta_i$, $\beta_a$ in the previous example \ref{ex:Realistic} are reasonable and somewhat realistic. Nevertheless, in order to better understand the role that asymptomatic carriers can carry it is actually more convenient to slightly exaggerate the values by having $\beta_a \gg \beta_e$. We shall do so in the next example which illustrates one other role that asymptomatic carries can have. Namely, that in an early stage, i.e. when $S$ is still large, having more asymptomatic carriers can cause a faster increase of the infected population anticipating and increasing the peak.

\begin{example}\label{ex:Realistic2}
	In this example we use the same constants of the previous example except for the values of $\beta_i$ and $\beta_a$ which we shall set to be $\beta_i=0.5$ and $\beta_a=7$. As explained before, the fact that $\beta_a$ is so big when compared to $\beta_i$ means not only that the disease is extremely contagious but also that those individuals which are showing symptoms are isolated and handled extremely carefully while the asymptotic ones are not. In figure \ref{fig:Comparison_Same_Sum2} we have plotted in black and green the simulation obtained when $a_a=0.1$ (and $a_i=0.9$) and $a_a=0.2$ (and $a_i=0.8$) respectively. 
	\begin{figure}[h]
		\centering
		\includegraphics[width=0.5\textwidth,height=0.35\textheight]{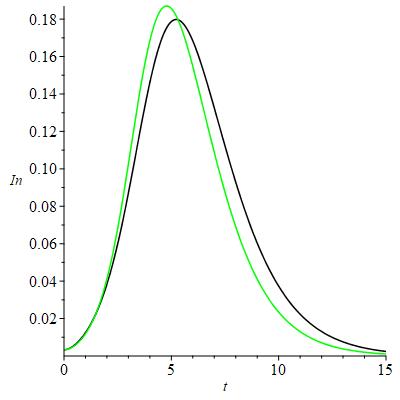}
		\caption{In black and green are the plots of $I(t)$ for the simulations obtained from $a_a$ being $0.1$ and $0.2$ respectively.}
		\label{fig:Comparison_Same_Sum2}
	\end{figure}
	In this very particular example, we see that the higher probability of an exposed person to become asymptomatic leads to higher degree of contagion which itself leads to a higher earlier peak of the infected population. It is natural to inquire about the actual population of asymptomatic carriers in both of these simulations, these are plotted in figure \ref{fig:Comparison_A}.
	\begin{figure}[h]
		\centering
		\includegraphics[width=0.5\textwidth,height=0.35\textheight]{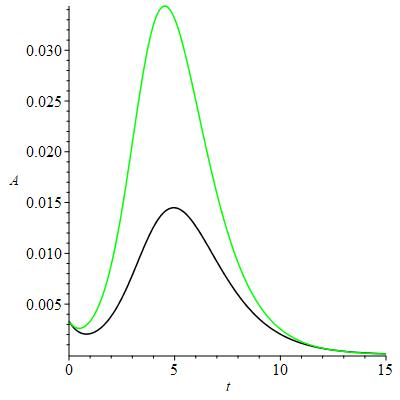}
		\caption{In black and green are the plots of $A(t)$ for $a_a$ being $0.1$, $0.2$ respectively.}
		\label{fig:Comparison_A}
	\end{figure}
\end{example}

\subsubsection{Second waves}

In the simpler model given by the system \ref{eq:ODE1}--\ref{eq:ODE4} we saw in Example \ref{ex:Second_Peak_Easy_Model} that a second wave can form. In the same manner as in that case, we expect this will happen when the quantity
$$\frac{1}{a_i+a_a} \left( \beta_e + a_i \frac{\beta_i}{\gamma_i + \mu} + a_a \frac{\beta_a}{\gamma_a} \right) ,$$
derived in subsection \ref{ss:Linear_Second_Model}, rises above $1$. We shall now see an example of a solution to the model \ref{eq:A1}--\ref{eq:A5} for which a second wave forms. In fact, we shall see that the second peak can be larger than the first and thus this example will serve as an answer to the items (b) and (c) of question \ref{que:2} using this more refined model. Also, we shall use this opportunity to further elaborate on the role of asymptomatic carriers giving further input towards question \ref{que:3} in the Introduction.

\begin{example}\label{ex:Second_Wave_Second_Model}
	Let $\beta_i(t)=0.9f(t)$, $\beta_e(t)=1.5f(t)$, $\beta_a(t)=2.8f(t)$, $a_a=0.2$ and $a_i=0.8$, $\gamma_i=0.8$, $\gamma_a=1.2$ and $\mu=0.01$, for some function $f(t)$ which was carefully designed\footnote{We are using $f(t)=\tfrac{1}{10}+S_{37}(2-t)+S_{37}(t-7)$ where $S_{37}(t)$ is the sum of the first 37 terms of the Fourier series of the Heaviside function on $[-20,20]$.} but whiose specific form is unimportant for now. As a matter of showing the importance of the quantity \ref{eq:R0_A} we shall plot this in figure \ref{fig:R0} which we can see has two large plateaus, well above $1$, connected by a region where it is substantially smaller. The simulation run for these values is presented in figure \ref{fig:Second_Peak_2} and is made with the initial conditions $S(0)=0.99$, $E(0)=0.01$, $I(0)=0=R(0)$.\\ 
	One sees that the first peak is too small for a sufficient fraction of the population to acquire immunity and the high increase of the rate of transmissions, here enconded in the quantity \ref{eq:R0_A} leads to a second peak of the outbreak.
	\begin{figure}[h]
		\centering
		\includegraphics[width=0.5\textwidth,height=0.35\textheight]{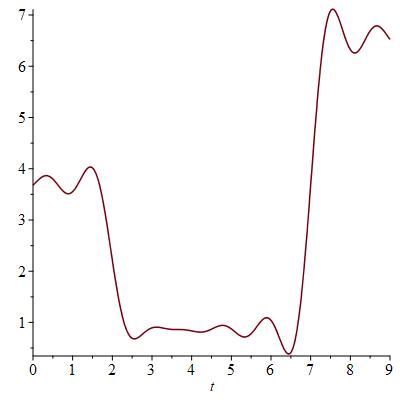}
		\caption{Plot of $\frac{1}{a_i+a_a} \left( \beta_e(t) + a_i \frac{\beta_i(t)}{\gamma_i + \mu} + a_a \frac{\beta_a(t)}{\gamma_a} \right)$.}
		\label{fig:R0}
	\end{figure}
	Indeed, this second wave of the outbreak creates an even larger peak than the first.
	\begin{figure}[h]
		\centering
		\includegraphics[width=0.5\textwidth,height=0.35\textheight]{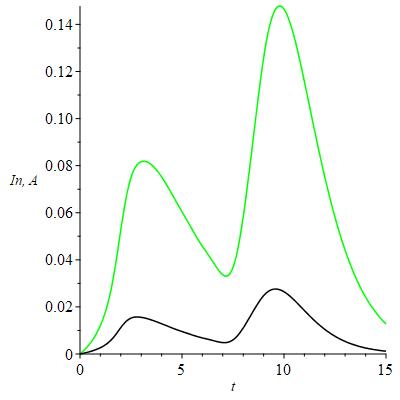}
		\caption{Plot of $I(t)$ and $A(t)$ in green and black. These represent the infected and the asymptomatic carriers respectively.}
		\label{fig:Second_Peak_2}
	\end{figure}
	Of course, the actual values of the transmission rates we used in this example may not be well adapted for a realistic modeling.\footnote{Even though, I have tried to make them as reasonable as I could.} Nevertheless, this serves as an example to stress the possibility of having a second outbreak that may be worse than the first.\\ 
	As a way of comparison, and to further investigate on the role of asymptomatic carriers, we present in figure \ref{fig:Second_Peak_2B} a simulation of the same model except that we now have a much higher fraction of asymptomatic carriers, in this case $a_a=0.8$ and $a_i=0.2$. This means that each individual has a $80$ per cent probability of becoming asymptomatic when exposed to the disease.
	\begin{figure}[h]
		\centering
		\includegraphics[width=0.5\textwidth,height=0.35\textheight]{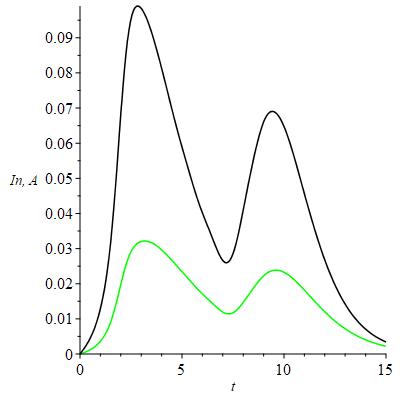}
		\caption{Plot of $I(t)$ and $A(t)$ in green and black respectively.}
		\label{fig:Second_Peak_2B}
	\end{figure}
	In this case, we see that the large quantity of asymptomatic carriers keeps the peaks of infection much smaller. In fact, we also find that in this case the second peak is smaller. This can be interpreted in terms of the herd immunity acquired by the asymptomatic carriers which shields the rest of population while keeping the number of infections smaller.
\end{example}

\begin{remark}
	Example \ref{ex:Second_Wave_Second_Model} raises the question of whether it is possibly to find conditions under which an outbreak will develop a unique peak. This is precisely the constant of item (a) in question \ref{que:2}. As we shall see in Corollary \ref{cor:One_Peak} of the next section such a criteria can be found and cast in terms of the quantity \ref{eq:R0_A}.
\end{remark}

\subsection{Precise qualitative results}

As in Lemma \ref{lem:E_I_to_zero} we find that in order for
$$S_c := \lim_{t \to + \infty}S(t),$$
to be positive we must have all $A(t)$, $I(t)$ and $E(t)$ converging to zero as $t \to + \infty$. The other natural easy question to investigate is that of finding criteria for which there will be only one peak of the outbreak, very much as done in Proposition \ref{lem:Only_One_Maximum} for the simpler model \ref{eq:ODE1}--\ref{eq:ODE4}. Here, we shall find an analogue of that which holds under more restrictive hypothesis.

\begin{proposition}\label{prop:One_Peak}
	Let $(S,E,A,I,R)$ be a solution to \ref{eq:A1}--\ref{eq:A5} with $a_i,a_a, \gamma_i, \mu$ constant and $S_c>0$. Then, there is a constant $c>0$ so that $A \leq c \frac{\gamma_i+\mu}{\gamma_a} \frac{a_a}{a_i} I$ for all $t \geq 0$. Furthermore, if
	$$\sup_{t \geq 0} \frac{1}{a_i +a_a}\left( \frac{a_i}{\gamma_i + \mu} \beta_i + \beta_e  +  c \frac{a_a}{\gamma_a}  \beta_a \right) \leq 1,$$
	the function $I(t)$ has at most one critical point and this is a maximum if it exists.
\end{proposition}
\begin{proof}
	First notice that as $S_c>0$ we have that both $A$ and $I$ converge to zero and so there is a constant $c$ as in the statement.
	Thus, in order to prove the result we shall follow the strategy of Proposition \ref{lem:Only_One_Maximum}. This consists in finding conditions which guarantee that $\ddot{I}<0$ at any critical point of $I$. If that is the case, then $I$ can have no minimum and thus will have at most one maximum, as in-between any two maxima there must be a minima. Thus, we start by computing
	\begin{align*}
	\ddot{I} & = a_i \dot{E} - (\gamma_i + \mu) \dot{I} \\
	& = a_i \left( \left( \beta_i I+\beta_e E+ \beta_aA \right)S - (a_i +a_a) E \right) - (\gamma_i + \mu) \dot{I} ,
	\end{align*}
	which at a critical point of $I$ further satisfies $\dot{I}=0$, i.e. $a_iE=(\gamma_i + \mu) I$, and so we can rewrite
	\begin{align*}
	\ddot{I} & = a_i \left( \left( \beta_i I+\frac{\beta_e}{a_i} (\gamma_i+\mu) I + \beta_aA \right)S - \frac{a_i +a_a}{a_i} (\gamma_i + \mu) I \right).
	\end{align*}
	Thus, we want to find conditions under which the following inequality holds
	\begin{equation}\label{eq:Inequality_Intermediate}
	\left( \beta_i +\frac{\beta_e}{a_i} (\gamma_i+\mu)  + \beta_a \frac{A}{I} \right)S - \frac{a_i +a_a}{a_i} (\gamma_i + \mu)  <0 .
	\end{equation}
	If we further have $A \leq c \frac{\gamma_i+\mu}{\gamma_a} \frac{a_a}{a_i} I$, then for inequality \ref{eq:Inequality_Intermediate} to hold, it is enough that
	$$
	\left( \beta_i +\frac{\beta_e}{a_i} (\gamma_i+\mu)  +  c \frac{\gamma_i+\mu}{\gamma_a} \frac{a_a}{a_i} \beta_a \right)S - \frac{a_i +a_a}{a_i} (\gamma_i + \mu)  <0 ,
	$$
	which we can rewrite as
	$$
	\frac{1}{a_i +a_a}\left( \frac{a_i}{\gamma_i + \mu} \beta_i + \beta_e  +  c \frac{a_a}{\gamma_a}  \beta_a \right)S  < 1.
	$$
	As $S \leq 1$ we find that this condition is immediate from that in the statement.
\end{proof}

\begin{remark}\label{rem:Only_One_Peak_A}
	From the proof of Proposition \ref{prop:One_Peak} we find that, in fact, for $I(t)$ to only have at most one maximum it is enough to show that
	$$ \sup_{t \geq 0} \frac{1}{a_i +a_a}\left( a_i\frac{\beta_i}{\gamma_i + \mu} \beta_i + \beta_e  + a_a \frac{\beta_a}{\gamma_i + \mu} \left( \frac{a_iA}{a_aI} \right) \right)  < 1 ,$$
	or even
	$$\sup_{t \geq 0}\frac{1}{a_i +a_a}\left( a_i\frac{\beta_i}{\gamma_i + \mu} \beta_i + \beta_e  + a_a \frac{\beta_a}{\gamma_i + \mu} \left( \frac{a_iA}{a_aI} \right) \right) S  < 1.$$
	Of course, one would like to keep $S(t)$ as close to one as possible and so the first of this seems to be a more reasonable test. Furthermore, it does not require to know the value of $S(t)$ but that of the ratio $A(t)/I(t)$ which may be obtained by sampling with anti-body tests for example. 
\end{remark}

Of course, the previous proposition would be much more useful if for any given solution $(S,E,A,I,R)$ one can determine the value of $c$.

\begin{lemma}\label{lem:A_I}
	Let $(S,E,A,I,R)$ be a solution to \ref{eq:A1}--\ref{eq:A5}. Then, for all $t \geq 0$ we have
	$$\frac{e^{-\gamma_a t}}{a_a}\frac{d}{dt} \left( e^{\gamma_a t} A \right) = \frac{e^{-(\gamma_i+\mu)t}}{a_i}\frac{d}{dt} \left( e^{(\gamma_i + \mu) t} I \right)$$
	Furthermore, let $A(0)$, $I(0)$ respectively denote the initial value of asymptomatic and infected individuals. If $a_iA(0)\leq a_aI(0)$ and $\gamma_a \geq \gamma_i + \mu$, then
	$$A(t) \leq \frac{a_a}{a_i}  I(t).$$
\end{lemma}
\begin{proof}
	Equations \ref{eq:A3}, \ref{eq:A4} can be equivalently rewritten as
	\begin{align*}
	\frac{d}{dt} \left( e^{\gamma_a t} A \right) & = a_a e^{\gamma_a t} E \\
	\frac{d}{dt} \left( e^{(\gamma_i + \mu) t} I \right) & = a_i e^{(\gamma_i+\mu) t} E ,
	\end{align*}
	from which we infer that
	\begin{align*}
	\frac{e^{-\gamma_a t}}{a_a}\frac{d}{dt} \left( e^{\gamma_a t} A \right) & = \frac{e^{-(\gamma_i+\mu)t}}{a_i}\frac{d}{dt} \left( e^{(\gamma_i + \mu) t} I \right) ,
	\end{align*}
	and using the hypothesis that $\gamma_a \geq \gamma_i + \mu$
	\begin{align*}
	\frac{1}{a_a}\frac{d}{dt} \left( e^{\gamma_a t} A \right) & = \frac{e^{(\gamma_a-(\gamma_i+\mu))t}}{a_i}\frac{d}{dt} \left( e^{\gamma_a t} I \right) \\
	& = \frac{1}{a_i}\frac{d}{dt} \left( e^{\gamma_a t} I \right) - \frac{\gamma_a-(\gamma_i+\mu)}{a_i} e^{\gamma_a t} I \\
	& \leq  \frac{1}{a_i}\frac{d}{dt} \left( e^{\gamma_a t} I \right) .
	\end{align*}
	Integrating this yields
	\begin{align*}
	\frac{e^{\gamma_a t} A(t)}{a_a} - \frac{A(0)}{a_a} \leq  \frac{e^{\gamma_a t} I(t)}{a_i} - \frac{I(0)}{a_i},
	\end{align*}
	which upon rearranging gives the result in the statement.
\end{proof}

Putting this Lemma \ref{lem:A_I} together with Proposition \ref{prop:One_Peak} we find the following Corollary. This serves as a rigorous answer to the inquire raised in item (a) of question \ref{que:2}.

\begin{corollary}\label{cor:One_Peak}
	Let $(S,E,A,I,R)$ be a solution to \ref{eq:A1}--\ref{eq:A5} with $a_i,a_a, \gamma_i, \mu$ constant and $S_c>0$. Suppose that $a_iA(0)\leq a_aI(0)$, $\gamma_a\geq \gamma_i + \mu$ and
	\begin{equation}\label{eq:Cor}
	\sup_{t \geq 0} \frac{1}{a_i +a_a}\left( \beta_e + \frac{a_i\beta_i + a_a \beta_a}{\gamma_i + \mu}   \right) < 1.
	\end{equation}
	Then, the function $I(t)$ has at most one critical point and this is a maximum if it exists.
\end{corollary}

\begin{remark}\label{rem:After_Cor}
	In the previous Corollary \ref{cor:One_Peak}, the hypothesis that $\gamma_a>\gamma_i + \mu$ holds immediately if the average time for an asymptomatic carrier to stop being contageous is smaller than that required by the average infected individual. Still under this hypothesis, if the condition in equation \ref{eq:Cor} holds, then we also have
	$$\sup_{t \geq 0} \frac{1}{a_i +a_a}\left( \beta_e + \frac{a_i\beta_i}{\gamma_i + \mu}  + \frac{a_a \beta_a}{\gamma_a} \right) < 1,$$
	which should be compared with Conclusion \ref{conclusion:Second_Model}.
\end{remark}

\section{Grouping the population}\label{sec:Groups}

Recall that the idea to be implemented and tested is that of splitting the population into $n$ groups which, to first approximation, do not interact. This obviously supposes that in each household there are only people of the same group as otherwise it would be impossible to guarantee that they do not mix.\\

We shall simply describe how the situation works by using the simplest possible model, namely the modified SEIR model described by equations \ref{eq:ODE1}--\ref{eq:ODE4}. For the SEIAR model the situation is similar and the argument follows exactly the same lines.\\
We start by splitting the fraction of susceptible population as $S=\sum_{k=1}^nS_k$ with similar splits for $E$, $I$ and $R$. Suppose that people from each group only interact with those of their group, i.e. that for all $i , j \in \lbrace 1, \ldots , k \rbrace$ and $i \neq j$, no person from group $i$ meets a person from group $j$. Then, each $k \in \lbrace 1 , \ldots , k \rbrace$ we have that $(S_k,E_k,I_k,R_k)$ solves

\begin{align*}
\dot{S}_k & =  - \beta_i I_k S_k - \beta_e E_k S_k  \\
\dot{E}_k & = \beta_i I_k S_k + \beta_e E_k S_k - a E_k \\
\dot{I}_k & = a E_k - \gamma I_k - \mu I_k \\
\dot{R}_k & = \gamma I_k .
\end{align*}
 
If we suppose that all populations follow the same dynamics with the same initial data then $S_1 = \ldots = S_n $ and so $S_k=S/n$ for each $k=1, \ldots , n$, and similarly $E_k=E/n$, $I_k=I/n$, $R_k=R/n$. Then, the system above turns into

\begin{align*}
\dot{S} & =  - \frac{\beta_i}{n} I S - \frac{\beta_e}{n} E S  \\
\dot{E} & = \frac{\beta_i}{n} I S  + \frac{\beta_e}{n} E S - a E \\
\dot{I} & = a E - \gamma I - \mu I \\
\dot{R} & = \gamma I ,
\end{align*}

This is the same system as that we started with but with the transmission rates $\beta$'s replaced by $\beta/n$. This cuts the rate of propagation by dividing it by the number of groups in which the population was split. 

\begin{remark}
	The exact same reasoning leads to a similar conclusion using the model \ref{eq:A1}--\ref{eq:A5}. 
\end{remark}

Building on our Conclusion \ref{conclusion:First_Model} we are lead to the following:

\begin{conclusion}
	It is possible to reduce the rates of transmission in half, or third, ou a fourth, and so on... Then, reducing it by a large enough factor, it is one can keep a large fraction of the population $S_c<S(0) <1$ still susceptible (without ever having contracted the disease). One way to achieve thus is to do divide the population into $n$ groups which one supposes not to physically interact with one another. Then, one keeps $n$ large enough so that
	$$\frac{1}{n} \frac{\beta_e}{\gamma + \mu}  < 1 .$$
	and
	$$\frac{1}{n} \left( \frac{\beta_i}{\gamma+\mu} + \frac{\beta_e}{a} \right) < 1.$$
	In fact, we should actually require that $\frac{\beta_e}{n}  \ll \gamma + \mu$ and $\frac{1}{n} \left( \frac{\beta_i}{\gamma+\mu} + \frac{\beta_e}{a} \right) \ll 1$, so that the final fraction of the population which is still susceptible is as high as possible.\\
	Using instead the model \ref{eq:A1}--\ref{eq:A5} and having in mind Corollary \ref{cor:One_Peak} and Remark \ref{rem:After_Cor}, if $\gamma_a \geq \gamma_i + \mu$ and $a_iA(0)\leq a_aI(0)$, then requiring that
	$$\sup_{t \geq 0}\frac{1}{a_i+a_a} \left( \beta_e + a_i \frac{\beta_i}{\gamma_i + \mu} + a_a \frac{\beta_a}{\gamma_a} \right) <1 ,$$
	guarantees that only one peak will form. Alternatively, if the ratio $A/I$ is possible to compute,\footnote{by using anti-body testing for instance.} then one may instead control the quantities put forward in remark \ref{rem:Only_One_Peak_A}.
\end{conclusion}

%There is an obvious problem with this model which the fact that it may be impossible to actually split the infected population into these groups because they will be either recovering at home or in the hospital. If they are at the hospital it will be virtually impossible/impractical to have them being treated by people in the same group.

\begin{example}\label{ex:Second}
	Consider the same numerical values as those we previously considered in example \ref{ex:First} and figure \ref{fig:First}. Recall that these are $\beta_i=0.9$, $\beta_e=2.5$, $a=1$, $\gamma=0.9$, $\mu=0.01$ and initial conditions $S(0)=0.99$, $I(0)=0$, $E(0)=0.01$, $R(0)=0$. Also, the curves with color red, blue, green and purple respectively denote the susceptible, exposed, infected and recovered. In figure \ref{fig:First2} the case when $n=2$ is plotted. This is a substantial improvement to the case of example \ref{ex:First}. Indeed, we see that about $30$ per cent of the population got away without ever being infected.
	\begin{figure}[h]
		\centering
		\includegraphics[width=0.4\textwidth,height=0.25\textheight]{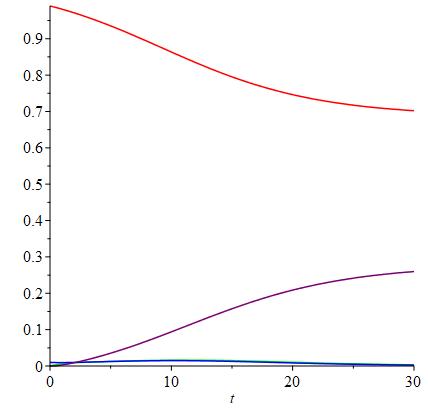}
		\caption{Example with $n=3$ and the color code of figures \ref{fig:First} and \ref{fig:First2}.}
		\label{fig:Third}
	\end{figure}
	Even better is the case when $n=3$ which is plotted in figure \ref{fig:Third} just as a further example. This represents an extremely good scenarium where the number of infected is kept very low and almost disappears by time $t=30$.
\end{example}

We shall now give an example using the more refined model from the system \ref{eq:A1}--\ref{eq:A5}.

\begin{example}
	Let $\beta_i=0.9$, $\beta_e=1.5$, $\beta_a=2.8$, while $a_a=0.2$, $a_i=0.8$ and $\gamma_i=0.8$, $\gamma_i=1.2$, $\mu=0.01$. Set the initial conditions to be $S(0)=0.99$, $E(0)=0.01$, $I(0)=0$, $R(0)=0$. Using, these values we run in figure \ref{fig:N2} two simulations, one for the case when $n=1$ which means the population was not split at all and one for $n=2$, i.e. the population was split into two groups.
	\begin{figure}[h]
		\centering
		\includegraphics[width=0.4\textwidth,height=0.25\textheight]{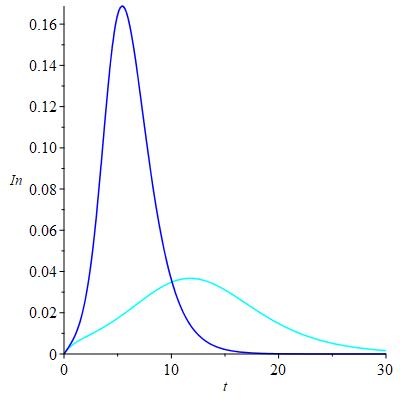}
		\caption{Comparing $I(t)$, the fraction of infected in the whole population, for $n=1$ in blue and $n=2$ in cyan.}
		\label{fig:N2}
	\end{figure}
	This suggests that, indeed, splitting the population into two groups leads to a much smaller peak. Nevertheless, it does take more time to eliminate the infection.
\end{example}


\begin{thebibliography}{5}	
	\bibitem[AR]{AR} Anderson, Roy M. and May, Robert M. \textit{Infectious diseases of humans: dynamics and control}. Oxford: Oxford University Press, (1991)
	
	\bibitem[Ae]{Ae} Arenas, Alex, et al. \textit{A mathematical model for the spatiotemporal epidemic spreading of COVID19}. medRxiv (2020)
	
	\bibitem[C]{C} Capasso, V. \textit{Mathematical structures of epidemic systems}. Vol. 97. Springer Science and Business Media, (2008)
	
	\bibitem[MK]{MK} McKendrick, A. G. and Kermack, W. O. \textit{A Contribution to the Mathematical Theory of Epidemics}. Proceedings of the Royal Society A. 115 (772): 700-721 (1927)
	
	\bibitem[Re]{Re} Li, Ruiyun, et al. \textit{Substantial undocumented infection facilitates the rapid dissemination of novel coronavirus (SARS-CoV2)}. Science, 16 Mar (2020)
	
	\bibitem[Vo]{Vo} Volz, Erik, et al. \textit{Genomic epidemiology of a densely sampled COVID19 outbreak in China}. medRxiv (2020).
	
	
\end{thebibliography}
\end{document}